\documentclass{amsart}

\usepackage{amsmath}
\usepackage{mathtools}
\usepackage{amssymb}
\usepackage{enumerate}
\usepackage{slashed}
\usepackage{graphicx}
\usepackage{graphbox}
\usepackage{caption}
\usepackage{subcaption}
\usepackage{newlfont}
\usepackage{amsrefs}
\usepackage{comment}
\usepackage[normalem]{ulem}
\usepackage{mathtools}
\usepackage{accents}
\usepackage{leftidx}
\usepackage{epsfig,fancyhdr,color}%,showkeys,amsmidx

\numberwithin{equation}{section}

\definecolor{NoteColor}{rgb}{1,0,0}

\theoremstyle{plain}
\newtheorem{theorem}{Theorem}%[section]
\newtheorem{lemma}[theorem]{Lemma}

\newtheorem{cor}[theorem]{Corollary}

\newtheorem{proposition}[theorem]{Proposition}

\theoremstyle{definition}

\theoremstyle{remark}

\newcommand{\R}{\mathbb{R}}

\newcommand{\del}{\partial}

\begin{document}

%\title{Topological Classification of Stationary Vacuum Black Holes with bi-axisymmetry in 5 %Dimensions}

\title[Plumbing Constructions and the Domain of Outer Communication]{Plumbing Constructions and the Domain of Outer Communication for 5-Dimensional Stationary Black Holes}

\author[Khuri]{Marcus Khuri}
\address{Department of Mathematics\\
Stony Brook University\\
Stony Brook, NY 11794, USA}
\email{khuri@math.sunysb.edu}

\author[Matsumoto]{Yukio Matsumoto}
\address{Department of Mathematics\\
Gakushuin University\\
Tokyo 171-8588, Japan}
\email{yukiomat@math.gakushuin.ac.jp}

\author[Weinstein]{Gilbert Weinstein}
\address{Physics Department and Department of Mathematics\\
Ariel University\\
Ariel, 40700, Israel}
\email{gilbertw@ariel.ac.il}

\author[Yamada]{Sumio Yamada}
\address{Department of Mathematics\\
Gakushuin University\\
Tokyo 171-8588, Japan}
\email{yamada@math.gakushuin.ac.jp}

\thanks{M. Khuri acknowledges the support of NSF Grant DMS-1708798. S. Yamada acknowledges the support of JSPS Grants KAKENHI 24340009 and 17H01091.}

\begin{abstract}
The topology of the domain of outer communication for 5-dimensional stationary bi-axisymmetric black holes is classified in terms of disc bundles over the 2-sphere and plumbing constructions. In particular we find an algorithmic bijective correspondence between the plumbing of disc bundles and the rod structure formalism for such spacetimes.
Furthermore, we describe a canonical fill-in for the black hole region and cap for the asymptotic region. The resulting compactified domain of outer communication is then shown to be homeomorphic to $S^4$, a connected sum of $S^2\times S^2$'s, or a connected sum of complex projective planes $\mathbb{CP}^2$. Combined with recent existence results, it is shown that all such topological types are realized by vacuum solutions. In addition, our methods treat all possible types of asymptotic ends, including spacetimes which are asymptotically flat, asymptotically Kaluza-Klein, or asymptotically locally Euclidean.
\end{abstract}
\maketitle

\section{Introduction}
\label{sec1} \setcounter{equation}{0}
\setcounter{section}{1}

In the classical 4-dimensional setting, the topology of horizon cross sections as well as the domain of outer communication for stationary asymptotically flat black holes is unique up to the number of horizon components, assuming appropriate energy conditions. Namely, Hawking's theorem \cite{Hawking,HawkingEllis} states that cross sections of the event horizon must be 2-spheres and topological censorship \cite{FriedmanSchleichWitt} combined with the positive resolution of Poincar\'{e}'s conjecture imply that the domain of outer communication must be the complement of a number of 3-balls in Euclidean space $\mathbb{R}\times \left(\mathbb{R}^3\setminus \cup_{i} B_{i}^3\right)$. In higher dimensions stationary black holes can have a variety of topologies for their horizon cross sections \cite{EmparanReall},
although each component must be of positive Yamabe type \cite{Galloway1} under proper energy assumptions. Moreover, relatively little is known about the domain of outer communication (DOC) \cite{HollandsIshibashi}. In this paper we restrict attention to the case of spacetime dimension 5. According to the Rigidity Theorem \cite{HollandsIshibashi0,HollandsIshibashiWald,MoncriefIsenberg}, generically a stationary solution must have at least one additional Killing field corresponding to a rotation. In fact, all known solutions in this dimension have two rotational symmetries, and we will therefore assume throughout that the symmetry group for the spacetime is $\mathbb{R}\times U(1)^2$. For such spacetimes satisfying the null energy condition, the list of possible horizon cross-sectional topologies is restricted to $S^3$, $S^1\times S^2$, and the lens spaces $L(p,q)$. Existence results for harmonic maps with prescribed singularities \cite{KhuriWeinsteinYamada,KhuriWeinsteinYamada1} have been applied to obtain vacuum solutions possessing each of these possible horizon topologies, in addition to various types of asymptotic structures, although the issue of (geometric) conical singularities is still open for the black lenses. The purpose of this current work is to classify the topologies of the DOC for these solutions, as well as those for other theories which respect the null energy condition.

Let $\mathcal{M}^5$ be the DOC of an orientable stationary bi-axisymmetric spacetime on which matter fields satisfy the null energy condition. It is also assumed that the stationary Killing field $\partial_\tau$ has complete orbits, and the DOC is globally hyperbolic having a Cauchy surface whose boundary is a compact cross-section of the event horizon.
Then $\mathcal{M}^5=\mathbb{R}\times M^4$ where the Cauchy surface $M^4$ is given by the $\tau=0$ slice. Various types of asymptotic ends will be considered, and their topology will be denoted by $M^4_{\text{end}}$. In particular the Cauchy surface may have an end which is asymptotically flat, asymptotically Kaluza-Klein, or asymptotically locally Euclidean which is homeomorphic to $\mathbb{R}_{+}\times S^3$, $\mathbb{R}_{+}\times S^1\times S^2$, or $\mathbb{R}_{+}\times L(p,q)$ respectively. Geometrically asymptotically cylindrical ends may also be present, as is the case with degenerate horizons. In this situation, as above, cross-sections of the cylindrical ends may take anyone of the three types of horizon topologies.

The orbit space $\mathcal{M}^{5}/[\mathbb{R}\times U(1)^2]$ is homeomorphic to the right-half plane $\{(\rho,z)\mid \rho\geq 0\}$ \cite{HollandsYazadjiev1} where the $z$-axis encodes nontrivial aspects of the topology. This result relies on the
topological censorship theorem \cite{ChruscielGallowaySolis,FriedmanSchleichWitt,
GallowaySchleichWittWoolgar,GallowaySchleichWittWoolgar1}, which in turn assumes
the null energy condition; it is for this reason that the null energy condition
is listed among the hypotheses in the current work.
The functions $\rho$ and $z$ are part of the global system of Weyl-Papapetrou coordinates $(\tau,\phi^1,\phi^2,\rho,z)$ which parameterize the DOC, where $\phi^a$ are $2\pi$-periodic and $\partial_{\phi^a}$, $a=1,2$ generate the $U(1)^2$ symmetry. The $z$-axis is divided into a sequence of intervals referred to as \textit{rods}
\begin{equation}
\Gamma_{1}=[z_{1},\infty),\text{ }\Gamma_{2}=[z_2,z_1],\text{ }\ldots,\text{ }
\Gamma_{L}=[z_{L},z_{L-1}],\text{ }\Gamma_{L+1}=(-\infty,z_{L}],
\end{equation}
and for each rod there is an associated \textit{rod structure} $(m_l,n_l)$ consisting of two integers having the property that the Killing field
\begin{equation}
m_l \partial_{\phi^1}+n_l\partial_{\phi^2}
\end{equation}
vanishes along $\Gamma_l$. If the rod structure $(m_l,n_l)=(0,0)$ then $\Gamma_l$ is called a \textit{horizon rod}, otherwise it is called an \textit{axis rod}. A point that separates two axis rods is a \textit{corner} if both generators of the Killing symmetry vanish there, otherwise it is a \textit{horizon puncture} and neither generator vanishes at that point. Horizon punctures may be taken to represent components of a degenerate horizon cross-section. In order to avoid orbifold singularities, the following condition is imposed on the determinant of neighboring rod structures surrounding a corner
\begin{equation}\label{det condition}
\begin{vmatrix}
m_l & m_{l+1} \\
n_l & n_{l+1} \end{vmatrix} = \pm 1.
\end{equation}
This ensures that a neighborhood of the corner in $M^4$ is homeomorphic to the 4-ball $B^4$ \cite{HollandsYazadjiev,KhuriWeinsteinYamada}.

It will be shown that certain neighborhoods of individual axis rods are topologically twisted disc bundles over the 2-sphere. Such fiber bundles will be denoted by $\xi$, and are classified by an integer $-k$ which represents the self-intersection number of the zero-section. These 4-manifolds are simply connected and have lens space boundary $\partial\xi=L(k,1)$ if $k\neq 0$. They will play the role of building blocks in the topological classification of the DOC. A consecutive sequence of axis rods gives rise to a neighborhood which may be identified with the 4-manifold obtained by plumbing together all of the associated disc bundles $\{\xi_{i}\}_{i=1}^{I}$. The plumbed manifold $\mathcal{P}(\xi_1,\cdots,\xi_{I})$ is again simply connected and has a lens space boundary $L(p,q)$, where $p$ and $q$ are computed in terms of the self-intersection numbers $\{-k_i\}_{i=1}^{I}$. Note that in terms of the plumbing notation $\mathcal{P}(\xi)=\xi$. The plumbing constructions will be described in detail in Section \ref{sec2}.

Our first main theorem provides a decomposition of the Cauchy surface into building blocks. These consist of plumbing constructions, the asymptotic end, 4-dimensional balls, and products of a disc with a cylinder.

\begin{theorem}\label{thm1}
The topology of the domain of outer communication of an orientable stationary bi-axisymmetric spacetime satisfying the null energy condition is $\mathcal{M}^5=\mathbb{R}\times M^4$ with the Cauchy surface given by a union of the form
\begin{equation}\label{1}
M^4=\cup_{j=1}^{J}\mathcal{P}\left(\xi_{1,j},\cdots,\xi_{I_j,j}\right)
\cup_{n=1}^{N_1}C_{n}^{4}\cup_{m=1}^{N_2}B_{m}^{4}\cup M_{\text{end}}^4,
\end{equation}
in which each constituent is a closed manifold with boundary and all are mutually disjoint expect possibly at the boundaries.
Each disc bundle $\xi_{i,j}$ is associated to an axis rod $\Gamma_{i,j}$ which is flanked on both sides by axis rods $\Gamma_{(i-1),j}$ and $\Gamma_{(i+1),j}$, $B_{m}^4$ is a 4-ball,  $C^{4}_{n}$ is $D^2\times S^1\times [0,1]$, and $M_{\text{end}}^4$ is either $\mathbb{R}_{+}\times S^3$, $\mathbb{R}_{+}\times S^1\times S^2$, or $\mathbb{R}_{+}\times L(p,q)$ depending on whether the spacetime is asymptotically flat, asymptotically Kaluza-Klein, or asymptotically locally Euclidean.
The value $J+N_2-1$ coincides with the number of connected components of the $z$-axis having at least one corner after horizon rods/punctures have been removed, $N_1$ is the number of single axis rods bounded by a horizon rod/puncture, and $N_2$ is the number of two consecutive axis rods which are bounded on either side by a horizon rod, horizon puncture, or the asymptotic end.  Moreover, the self-intersection number of the zero-section for the disc bundle $\xi_{i,j}$ is computed by
\begin{equation}\label{1.1}
-k_{i,j}=
\begin{vmatrix}
m_{(i-1),j}  & m_{i,j}     \\
 n_{(i-1),j}  & n_{i,j}    \\
\end{vmatrix}
\begin{vmatrix}
m_{i,j}  &  m_{(i+1),j}     \\
n_{i,j}    &  n_{(i+1),j} \\
\end{vmatrix}
\begin{vmatrix}
 m_{(i+1),j}  &  m_{(i-1),j}     \\
n_{(i+1),j}    &  n_{(i-1),j}  \\
\end{vmatrix},
\end{equation}
where $(m_{i,j},n_{i,j})$ denotes the rod structure for $\Gamma_{i,j}$.
\end{theorem}

While this result identifies the fundamental constituents of the DOC along with an algorithmic method for computing them, it does not express the topology in a concise way. A simplified expression may be obtained by filling in the horizons with canonically chosen simply connected compact 4-manifolds, and similarly capping off the asymptotic end to obtain a compactified manifold without boundary. Since this manifold is simply connected, the work of Freedman \cite{FreedmanQuinn} and Donaldson \cite{DonaldsonKronheimer} yields a classification of the `compactified DOC'. The procedure for filling in a horizon or capping off an asymptotic end is algorithmic as well, and consists of the plumbing of a finite number of disc bundles over $S^2$. This plumbing construction is naturally associated with a set of subrod structures for rods which may be thought of as existing within the black hole region or at infinity. The disc bundles used to fill in a particular horizon or end are determined by a continued fraction expansion arising from the two rod structures bounding the horizon rod/puncture or end, in that elements of the continued fraction are precisely the self-intersection numbers for the disc bundles. Furthermore, from these self-intersection numbers the desired rod structures may be computed inductively.

\begin{theorem}\label{thm2}
Consider the domain of outer communication $\mathcal{M}^5=\mathbb{R}\times M^4$ of an orientable stationary bi-axisymmetric spacetime satisfying the null energy condition, with $H$ horizon cross-sectional components. There exists a choice of horizon fill-ins $\{\tilde{M}_{h}^{4}\}_{h=1}^{H}$ and a cap for the asymptotic end $\tilde{M}_{\text{end}}^{4}$, each of which is either a 4-ball $B^4$ or a plumbed finite sequence of disc bundles over the 2-sphere $\mathcal{P}(\xi_{h_1},\cdots,\xi_{h_I})$, such that the compactified Cauchy surface
\begin{equation}
\tilde{M}^4=\left(M^4\setminus M^{4}_{\text{end}}\right)\cup_{h=1}^{H}\tilde{M}_{h}^{4}\cup
\tilde{M}^{4}_{\text{end}}
\end{equation}
is homeomorphic to the sphere $S^4$, a connected sum of 2-sphere products $\# m S^2 \times S^2$, or a connected sum of complex projective planes $\left(\# n\mathbb{CP}^2\right) \#
\left(\# \ell \overline{\mathbb{CP}}^2\right)$.
Moreover, the disc bundles for each fill-in and cap may be computed algorithmically from the neighboring rod structures of each horizon and the asymptotic end.
\end{theorem}

This may be considered a direct generalization of the corresponding statement in $D=4$ given in the first paragraph, where the compactified space is $S^3$.
A similar result was established by Hollands et al. \cite{HollandsHollandIshibashi,HollandsYazadjiev1} in the asymptotically flat and asymptotically Kaluza-Klein cases with nondegenerate horizons.  Their version of the compactified manifold $\tilde{M}^4$ is classified topologically as either $S^4$ or $\left(\# m S^2 \times S^2\right) \# \left(\# n\mathbb{CP}^2\right) \#
\left(\# \ell \overline{\mathbb{CP}}^2\right)$.  Here $\overline{\mathbb{CP}}^2$ is the complex projective plane with opposite orientation to $\mathbb{CP}^2$.
Therefore Theorem \ref{thm2} may be considered as a refinement of their result.
%Note that the complex projective spaces may be eliminated from the list of
%possibilities if it is assumed that the spacetime is spin.
In addition, it should be pointed out that our method for filling in horizons is different from that in \cite{HollandsHollandIshibashi,HollandsYazadjiev1}, since for instance we obtain different compactified DOCs for the single component black ring. Namely, the procedure of \cite{HollandsHollandIshibashi,HollandsYazadjiev1} produces $S^4$ whereas our method yields $S^2\times S^2$ for $\tilde{M}^4$ in the case of asymptotically flat black rings. This example and others will be described in detail in Section \ref{sec4}. Furthermore an important contribution of Theorem \ref{thm2}, which separates it from previous results, is the introduction of an algorithm for computing the topology of the DOC. Finally, we note that simple connectivity of the compactified DOC is consistent with
topological censorship \cite{ChruscielGallowaySolis,FriedmanSchleichWitt,
GallowaySchleichWittWoolgar,GallowaySchleichWittWoolgar1}.

It is a natural question to ask, which of the topologies for the compactified manifold $\tilde{M}^4$ described in Theorem \ref{thm2} can be realized by stationary vacuum solutions. Previously, very few examples were known.
In fact, in \cite[pg. 18]{HollandsIshibashi} it was commented that in all known solutions only $S^4$ arises. However, recently progress has been made with regards to the existence question for bi-axisymmetric solutions of the stationary vacuum equations having a variety of asymptotic ends. In \cite{KhuriWeinsteinYamada,KhuriWeinsteinYamada1}, existence results for harmonic maps with prescribed singularities have been utilized to construct bi-axisymmetric stationary vacuum spacetimes in 5-dimensions having arbitrary rod structures modulo
mild compatibility conditions. In particular, combining these existence results with Theorem \ref{thm2} answers the question posed above.

\begin{theorem}\label{thm3}
Each of the topologies listed in Theorem \ref{thm2} for the compactified Cauchy surface $\tilde{M}^4$ is realized by a solution of the 5-dimensional bi-axisymmetric stationary vacuum Einstein equations.
\end{theorem}

The solutions produced in \cite{KhuriWeinsteinYamada,KhuriWeinsteinYamada1} are given in terms of abstract existence results, and it is not immediately clear which of these solutions are absent of conical singularities. It is known, however, that conical singularities are not present on the two semi-infinite rods. Although the issue of conical singularities is relevant for physics and geometry, it plays no role in the topological classification. In particular, we conjecture that
any of the solutions produced in \cite{KhuriWeinsteinYamada,KhuriWeinsteinYamada1} can be perturbed to give smooth initial data, devoid of any conical singularity, with the same outermost apparent horizon topology. It would be of interest to analyze the DOCs for the evolutions of such data.

A basic question posed in the literature \cite{AlaeeKunduriPedroza} is the following. Does the topology of horizon cross-sections and the asymptotic end uniquely determine the topology of the domain of outer communication for stationary (vacuum) black holes in 5-dimensions? As a consequence of Theorem 3 we are able to answer this question. An example illustrating the answer is given in Section \ref{sec4}.

\begin{cor}
The topology of the domain of outer communication of a 5-dimensional stationary vacuum bi-axisymmetric black hole is not uniquely determined by the horizon cross-sectional topology and the topology of the asymptotic end. In particular, there exist two asymptotically flat black $\mathbb{RP}^3$'s having topologically different DOCs.
\end{cor}

%14) Don't forget that our results apply to the minimal supergravity case as well, and thus %can be used to describe the Kunduri-Lucietti solutions, as well as the Tomizawa-Nozawa %solutions.

\section{Plumbing Constructions}
\label{sec2} \setcounter{equation}{0}
\setcounter{section}{2}

Consider a disc bundle $\pi: \xi \rightarrow S^2$ over the 2-sphere whose zero-section has self-intersection number $-k\in \mathbb{Z}$.  Such a bundle can be constructed by gluing two trivial disc bundles $\pi^\pm: \xi^\pm \rightarrow D^\pm$ along the solid tori $\del{D^\pm} \times D^2$. Here $S^2=D^{+}\cup D^{-}$ is the union of the northern and southern hemisphere.  The gluing map
$f: \del {D^+} \times D^2 \rightarrow \del{D^-} \times D^2$
is given by
\begin{equation}
(z, v) \xmapsto{f}  \left(\overline{z}, e^{ik{\theta_0}}v= |v| e^{i(\varphi_0+ k\theta_0)}\right), 
\end{equation}
where $\arg(z) = \theta_0$ and $\arg (v) = \varphi_0$.
Note that the orientations of $\partial D^+$ and $\partial D^-$ induced by the natural orientation of $S^2 = D^+ \cup D^-$ are opposite of each other.  The conjugation $\overline{z}$ of the image of $f$ is introduced to reflect this fact. We write the effect of $f$ simply as
\begin{equation}
(\theta_0, \varphi_0) \xmapsto{f} (-\theta_0, \varphi_0+k\theta_0),
\end{equation}
where $f( \theta_0, \varphi_0) =: ( \theta_1, \varphi_1)$ with $e^{i \theta_1} \in \partial D^-$ and $v= |v|e^{i \varphi_1} \in \pi^{-1} (e^{i \theta_1})$. Observe that in the disc bundle $\xi$ there is a natural 2-torus action which rotates the base and fiber.  Moreover, the boundary of the total space $\partial\xi$ is homeomorphic to the lens space $L(k, 1)$ \cite{Orlik}, and according to van Kampen's theorem $\xi$ is simply connected. (Our presentation faithfully follows that of  \cite{Orlik}, with the sole difference being that our ``$k$" is their ``$-m$" in p.25.)

\begin{figure}
\includegraphics[width=6.5cm]{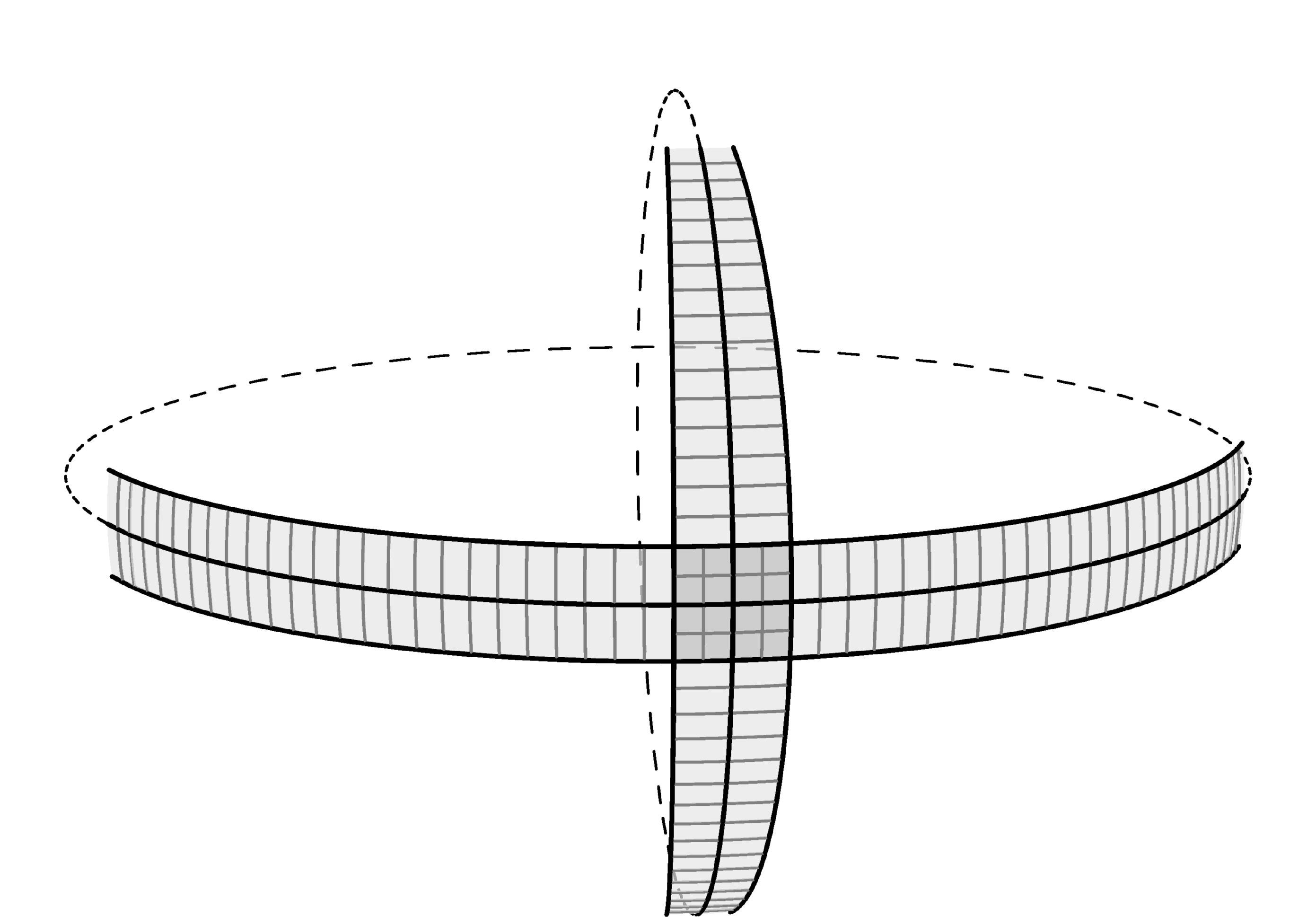}
\caption{Plumbing of two disc bundles}  \label{Picture}
\end{figure}

Two disc bundles $\xi_1$ and $\xi_2$ can be combined via an operation known as \textit{plumbing}.  First take a closed disc  $U_1$ centered at the origin of $D^-_1$ and another disc $U_2$ centered at the origin of $D^+_2$, where the bundle over $U_i$ is trivial. Next identify the pair of polydiscs $\pi^{-1}(U_1) \cong U_1 \times D^2 $ and $\pi^{-1}(U_2) \cong U_2 \times D^2 $ by interchanging fiber and base
\begin{equation}
(z, v) \sim (v, z)
\end{equation}
where  $(z, v) \in U_1 \times D^2 $ and $(v, z) \in U_2 \times D^2$. An illustration is given in Figure \ref{Picture}.
We denote the resulting 4-dimensional manifold with boundary by $\mathcal{P}(\xi_{1},\xi_{2})$, and note that it is simply connected by van Kampen's theorem.  If the first disc bundle $\xi_1$ is obtained by a gluing map $f_1$ so that the self-intersection number of the zero section is $-k_1$, and $\xi_2$'s gluing map $f_2$ induces the self-intersection number $-k_2$, then the boundary of  $\mathcal{P}(\xi_{1},\xi_{2})$ is homeomorphic to a lens space
$L(k_1k_2 - 1, k_2)$ \cite{Orlik}.  When $k_1, k_2 >1$, we note (\cite{Orlik}) that $k_1$ and $k_2$ determines a continued fraction
\begin{equation}
k_1 -\frac{1}{k_2} = \frac{k_1k_2 -1}{k_2}.
\end{equation}

By induction, one can construct a 4-dimensional simply connected manifold $\mathcal{P}(\xi_{1},\cdots,\xi_{\ell})$ with its boundary homeomorphic to a lens space $L(p, q)$, $p> q > 0$, by plumbing a sequence of disc bundles $\xi_i$ with base $S^2$ having self-intersection numbers $-k_i $ satisfying
\begin{equation}
\frac{p}{q} =  k_1- \cfrac{1}{k_2-\cfrac{1}{k_3-\cfrac{1}{{\cdots} - \frac{1}{k_\ell}}}}.
\end{equation}
Recall that each such rational number has a unique expansion of this form with $k_\ell>1$.

This construction has a canonical underlying $U(1)^2$-action, which will now be made more explicit. Consider a sequence of plumbed disc bundles $\mathcal{P}(\xi_1,\cdots,\xi_{\ell})$.
Let the $U(1)$-action around the origin of $D^+_1$ be represented by $t \in [0, 2\pi)$ and the $U(1)$-action on the fibers $D^2$ over $D^+_1$ by $s \in [0, 2 \pi)$, that is
in local coordinates over $D^+_1$ this action may be represented by
\begin{equation}
(\theta_0, \varphi_0) \xmapsto{(t, s)} (\theta_0 + t, \varphi_0 + s).
\end{equation}
Observe that the flow fields $\partial_t$ and $\partial_s$ agree with the coordinate
fields $\partial_{\theta_0}$ and $\partial_{\varphi_0}$, respectively. By working downward through the plumbing construction the $U(1)^2$-action may be described inductively at each stage in terms of these original flow fields.

Recall that the disc bundle $\xi_1$ over $S^2$ is a twisted union of two trivial bundles
\begin{equation}
[D^+_1 \times D^2] \cup_{f_1} [D^-_1 \times D^2]
\end{equation}
where $f_1: \del {D^+_1} \times D^2 \rightarrow \del{D^-_1} \times D^2$ is the gluing map.
The $U(1)^2$ action on $D^-_1 \times D^2$ is twisted by $f_1$ in the sense that
\begin{equation}
(-\theta_0, \varphi_0 +k_1 \theta_0 ) \xmapsto{(t, s)} (-\theta_0 - t, \varphi_0+k_1(\theta_0 +t) + s),
\end{equation}
which may be rewritten as
\begin{equation}
(\theta_1, \varphi_1)   \xmapsto{(t, s)} (\theta_1 - t, \varphi_1 + k_1 t + s)
\end{equation}
with the change of variables
$(\theta_1, \varphi_1) := f_1(\theta_0, \varphi_0) = (-\theta_0, \varphi_0+ k_1 \theta_0)$.
The value $-k_1$ is as above the self-intersection number of the `core curve' $S^2$, namely the zero-section of the disc bundle $\xi_1$. The relation between the flow and coordinate fields is then
\begin{equation}
\label{3-consecutive-rods-IV}
\begin{pmatrix}
    \partial_{\theta_1}   &  \partial_{\varphi_1}\\
\end{pmatrix}
=
\begin{pmatrix}
    \partial_{t}   &  \partial_{s}\\
\end{pmatrix}
\begin{pmatrix}
     -1    & 0\\
     k_1  &  1 \\
\end{pmatrix}.
\end{equation}
Here we note that the first column $(-1, k_1)^t$ induces $\frac{k_1}{-1}$, the so-called Dehn surgery coefficient, which in turn determines the lens space $L(k_1, 1)$  (see example 5.3.2 of \cite{GompfStipsicz}). The second column $(0,1)^t$ consists of rod structure coefficients. In addition, for the sake of clarification we also mention that the preceding disc bundle $\xi_{0}$ above $\xi_1$ has the $U(1)^2$-action
\begin{equation}
(\varphi_0, \theta_{0})  \xmapsto{(t, s)} (\varphi_0+s, \theta_{0}+t ),
\end{equation}
which induces
\begin{equation}
\label{3-consecutive-rods-III}
\begin{pmatrix}
    \partial_{\varphi_0}   &  \partial_{\theta_0}\\
\end{pmatrix}
=
\begin{pmatrix}
    \partial_{t}   &  \partial_{s}\\
\end{pmatrix}
\begin{pmatrix}
     0    & 1\\
     1  &   0 \\
\end{pmatrix}.
\end{equation}

In the next step of the induction process consider the plumbing of $\xi_2$ to $\xi_1$. This involves an identification between $D^-_1 \times D^2$ and $D^+_2 \times D^2$ in which the relevant local coordinates on $\xi_2$ are $(\varphi_1, \theta_1)$. That is, $\varphi_1$ is the argument of the base space $D^+_2$ and $\theta_1$ is the argument of the disc fiber. The $U(1)^2$-action on $D^+_2 \times D^2$ is inherited from $D^-_1 \times D^2$ as follows
\begin{equation}
(\varphi_1, \theta_1)   \xmapsto{(t, s)} ( \varphi_1 + k_1 t + s, \theta_1 - t).
\end{equation}
The $U(1)^2$-action on $D^-_2 \times D^2$ is twisted by $f_2$ and is given by
\begin{equation}
(- \varphi_1, \theta_1 + k_2 \varphi_1) \xmapsto{(t, s)} (-\varphi_1 - k_1 t - s,  \theta_1 + k_2(\varphi_1 + k_1 t + s) - t),
\end{equation}
which may be rewritten as
\begin{equation}
(\varphi_2, \theta_2)   \xmapsto{(t, s)} (\varphi_2 - k_1 t - s, \theta_2 + k_1k_2 t - t + k_2 s )
\end{equation}
with the change of variables
$(\varphi_2, \theta_2) := f_2(\varphi_1, \theta_1) = (-\varphi_1, \theta_1 + k_2 \varphi_1)$.
The following relation then holds between the flow and coordinate vector fields
\begin{equation}
\label{3-consecutive-rods-I}
\begin{pmatrix}
    \partial_{\varphi_2}  &  \partial_{\theta_2}\\
\end{pmatrix}
=
\begin{pmatrix}
    \partial_{t}   &  \partial_{s}\\
\end{pmatrix}
\begin{pmatrix}
     -k_2    & -1\\
      k_1k_2 -1  &  k_1 \\
\end{pmatrix}.
\end{equation}
The first column $(-k_2, k_1k_2 -1)^t$ induces the Dehn surgery coefficient $\frac{k_1k_2 -1}{-k_2}$, which in turn determines the lens space $L(k_1k_2-1, k_2)$, and the second column $(-1, k_1)^t$ is the rod structure that is inherited from the Dehn surgery coefficient of the previous bundle $\xi_1$.
Observe that this gives  rise to the continued fraction
\begin{equation}\label{two disc}
\frac{k_1 k_2 -1 }{k_2} = k_1 - \frac{1}{k_2},
\end{equation}
and the boundary $\partial\mathcal{P}(\xi_{1},\xi_{2})$ is the lens space $L(k_1 k_2-1,k_2)$.

Let us now plumb the third disc bundle $\xi_3$, with self-intersection number $k_3$, to the bottom of $\mathcal{P}(\xi_{1},\xi_{2})$ to obtain $\mathcal{P}(\xi_{1},\xi_{2},\xi_{3})$.  Recall that the core curve $S^2$ of $\xi_3$ is the union of the northern and southern hemispheres $D^+_3 \cup D^-_3$. The $U(1)^2$-action on $D^+_3\times D^2$ is inherited from the action on $D^-_2\times D^2$ by
\begin{equation}
(\theta_2, \varphi_2)  \xmapsto{(t, s)}  (\theta_2 + k_1k_2 t - t + k_2 s, \varphi_2 - k_1 t - s).
\end{equation}
The $U(1)^2$-action on $D^-_3\times D^2$ is then twisted by $f_3$ so that
\begin{equation}
(-\theta_2, \varphi_2+k_3 \theta_2)  \xmapsto{(t, s)} (-\theta_2+ (1-k_1k_2)t - k_2s, \varphi_2 +k_3\{\theta_2 + (k_1k_2 t-1)t+k_2 s\} - k_1 t -s),
\end{equation}
which is written as
\begin{equation}
(\theta_3, \varphi_3)  \xmapsto{(t, s)} (\theta_3 + (1-k_1k_2)t - k_2s, \varphi_3 + (k_1k_2k_3 - k_1-k_3)t + (k_2k_3-1)s )
\end{equation}
with the change of variables $(\theta_3, \varphi_3) := f_3(\theta_2, \varphi_2) = (-\theta_2, \varphi_2+k_3 \theta_2  )$. The relation between the flow and coordinate vector fields is then
\begin{equation}
\label{3-consecutive-rods-II}
\begin{pmatrix}
    \partial_{\theta_3}  &  \partial_{\varphi_3}\\
\end{pmatrix}
=
\begin{pmatrix}
    \partial_{t}   &  \partial_{s}\\
\end{pmatrix}
\begin{pmatrix}
     1- k_2 k_3    & - k_2\\
     k_1k_2k_3 -k_1 - k_3   &  k_1k_2-1 \\
\end{pmatrix}.
\end{equation}
This gives rise to the continued fraction for $k_i > 0$
\begin{equation}
\frac{ k_1k_2k_3 -k_1 - k_3  }{k_2k_3-1} = k_1  - \cfrac{1}{k_2-\cfrac{1}{k_3}},
\end{equation}
and the boundary of $\mathcal{P}(\xi_{1},\xi_{2},\xi_{3})$ is the lens space $L(k_1k_2k_3 -k_1 - k_3 , k_2k_3-1 )$.

This process may be continued inductively. The resulting $2 \times 2$ matrix representing the $U(1)^2$-symmetry of $\mathcal{P}(\xi_{1},\xi_{2}, \dots \xi_\ell)$ in terms of $\partial_{t}$ and $\partial_{s}$ encodes topological information about $\partial \mathcal{P}(\xi_{1},\xi_{2}, \dots \xi_\ell)$ in its first column via the Dehn coefficient, and in the second column it encodes the $U(1)$-symmetry of disc fibers for the $\ell$-th disc bundle $\xi_\ell$. We will see later that the second column of this $2 \times 2$ matrix is of particular importance for the relation with rod structures of stationary bi-axisymmetric black holes. When $k_i >0$, the inductive construction is associated to the following arithmetic algorithm.

\begin{proposition}\label{prop1}
Let $\{\xi_i\}_{i=1}^{\ell+1}$ be a sequence of disc bundles over $S^2$ with zero-section self-intersection numbers $-k_i$. Let $\partial_{t}$ and $\partial_{s}$ denote generators of the $U(1)^2$-action on the plumbing construction $\mathcal{P}(\xi_1,\cdots,\xi_{\ell+1})$ which coincide with the canonical rotation of base and fiber on the trivialization $D_1^+\times D^2$.
If each $k_i>0$ and $k_\ell>1$, then the $U(1)$-action on the disc fiber over $D^-_{\ell+1}$ is given by $-n_\ell\partial_t+m_\ell\partial_s$ for some $m_\ell,n_\ell\in\mathbb{Z}$ satisfying
\begin{equation}
\frac{m_\ell}{n_\ell} = k_1- \cfrac{1}{k_2-\cfrac{1}{k_3-\cfrac{1}{{\cdots} - \frac{1}{k_\ell}}}}.
\end{equation}
Furthermore, the boundary of the plumbed disc bundles $\partial\mathcal{P}(\xi_1,\cdots,\xi_{\ell+1})$ is diffeomorphic to the lens space $L( m_\ell,n_\ell)$.
\end{proposition}

The discussion above may be given a localized description which elucidates the connection between self-intersection numbers of zero-sections and the $U(1)$-action on disc fibers. The general form of the $U(1)^2$-action parameterized by $(t, s)$ on the trivialization $D^{-}_{i-1} \subset D^2$ within the bundle $\xi_{i-1}$ takes the form
\begin{equation}
(\varphi_{i-1}, \theta_{i-1}) \xmapsto{(t, s)} (\varphi_{i-1} + n_i t+ m_i s, \theta_{i-1} + q_i t + p_i s),
\end{equation}
where the integer coefficients satisfy the normalization condition $m_i q_i -n_i p_i =1$ so that the flows are diffeomorphisms. It follows that the relation between coordinate fields and generators of the action is
\begin{equation}\label{a}
\begin{pmatrix}
    \partial_{\varphi_{i-1}}   &  \partial_{\theta_{i-1}}\\
\end{pmatrix}
=
\begin{pmatrix}
    \partial_{t}   &  \partial_{s}\\
\end{pmatrix}
\begin{pmatrix}
-p_i    & m_i  \\
 q_i   &   -n_i  \\
\end{pmatrix}.
\end{equation}
This action is transmitted to the trivialization $D^+_i \subset D^2$ within $\xi_i$ as
\begin{equation}
(\theta_{i-1}, \varphi_{i-1}) \xmapsto{(t, s)} (\theta_{i-1} + q_i t + p_i s, \varphi_{i-1} +  n_i t + m_i s).
\end{equation}
Recall that the disc bundle $\xi_i$ is a union of two trivial bundles
\begin{equation}
[D^+_i \times D^2] \cup_{f_i} [D^-_i \times D^2].
\end{equation}
The twisting imposed by the gluing map $f_i$ yields the following expression for the action over $D^-_i$
\begin{equation}
(\theta_i, \varphi_i)   \xmapsto{(t, s)} (\theta_i - q_i t - p_i s, \varphi_i + (n_i+ k_i q_i) t + (m_i +k_i p_i) s ),
\end{equation}
where the change of variables is given by
$(\theta_i, \varphi_i) := f_i (\theta_{i-1}, \varphi_{i-1}) = (-\theta_{i-1}, \varphi_{i-1} + k_i \theta_{i-1})$. We then have
\begin{equation}\label{b}
\begin{pmatrix}
    \partial_{\theta_{i}}   &  \partial_{\varphi_{i}}\\
\end{pmatrix}
=
\begin{pmatrix}
    \partial_{t}   &  \partial_{s}\\
\end{pmatrix}
\begin{pmatrix}
    -m_i - k_i p_i    &  -p_i\\
     k_i q_i + n_i &   q_i\\
\end{pmatrix}.
\end{equation}
By continuing this algorithm, the desired formula for coordinate fields on the trivialization $D^{-}_{i+1}\times D^2$ is found to be
\begin{equation}\label{c}
\begin{pmatrix}
    \partial_{\varphi_{i+1}}   &  \partial_{\theta_{i+1}}\\
\end{pmatrix}
=
\begin{pmatrix}
    \partial_{t}   &  \partial_{s}\\
\end{pmatrix}
\begin{pmatrix}
  k_{i+1 } (-k_i p_i- m_i) + p_i    & -m_i - k_i p_i  \\
  k_{i+1} (n_i + k_i q_i) - q_i   &    k_i q_i + n_i \\
\end{pmatrix},
\end{equation}
where
\begin{equation}
(\varphi_{i+1}, \theta_{i+1}) := f_{i+1} (\varphi_i, \theta_i) = (\varphi_i, \theta_i-k_{i+1} \varphi_{i}).
\end{equation}
This generalizes the $U(2)^2$-action demonstrated in Equation (\ref{3-consecutive-rods-I}), where $i=1$ and $m_1=q_1=1, n_1=p_1 =0$.

The $U(1)$-action on disc fibers within the bundles $\xi_{i-1}$, $\xi_i$, and $\xi_{i+1}$ may now be read off from \eqref{a}, \eqref{b}, and \eqref{c}, and expressed in terms of the action generators as
\begin{equation}\label{d}
m_i\partial_t - n_i \partial_s,\quad\quad
-p_i\partial_t + q_i \partial_s,\quad\quad
-(m_i + k_i p_i)\partial_t +( k_i q_i + n_i)\partial_s.
\end{equation}
The self-intersection number of the zero-section within $\xi_i$ may now be computed as a product of determinants involving these vectors
\begin{equation}\label{e}
- k_i=
\begin{vmatrix}
 m_i  & -p_i     \\
 - n_i  & q_i    \\
\end{vmatrix}
\begin{vmatrix}
-p_i  &     -(m_i + k_i p_i)    \\
q_i    &   k_i q_i + n_i \\
\end{vmatrix}
\begin{vmatrix}
 -(m_i + k_i p_i) & m_i  \\
   k_i q_i + n_i & -n_i \\
\end{vmatrix}.
\end{equation}
This fact is relevant to the setting of stationary bi-axisymmetric spacetimes since in various applications knowledge of the action on disc fibers \eqref{d} will be given, and formula \eqref{e} allows one to then compute the self-intersection numbers from this data.

\begin{proposition}\label{prop2}
Consider a consecutive sequence of three disc bundles $\xi_{i-1}$, $\xi_i$, and $\xi_{i+1}$ within the plumbing construction $\mathcal{P}(\xi_1,\cdots,\xi_{\ell+1})$, such that the respective $U(1)$-actions on their fibers are given by
\begin{equation}
m_{i-1}\partial_t +n_{i-1}\partial_s, \quad\quad  m_{i}\partial_t+ n_{i}\partial_s, \quad\quad m_{i+1}\partial_t + n_{i+1}\partial_s.
\end{equation}
Then the self-intersection number of the zero-section of $\xi_i$ is given by
\begin{equation}
\begin{vmatrix}
m_{i-1}  & m_i     \\
 n_{i-1}  & n_i    \\
\end{vmatrix}
\begin{vmatrix}
m_i  &  m_{i+1}     \\
n_i    &  n_{i+1} \\
\end{vmatrix}
\begin{vmatrix}
 m_{i+1}  &  m_{i-1}     \\
n_{i+1}    &  n_{i-1}  \\
\end{vmatrix}.
\end{equation}
\end{proposition}

We note that a similar formula appears in p.544 of \cite{OrlikRaymond1}  with a sign difference, due to the difference in coordinates. Recall that our coordinate system is in accordance with \cite{Orlik}.

\section{Proof of the Main Theorems}
\label{sec3} \setcounter{equation}{0}
\setcounter{section}{3}

Consider a spacetime $\mathcal{M}^5$ as given in Theorem \ref{thm1}, with Cauchy surface $M^4$. The orbit space $M^4/U(1)^2$ is expressed as the $\rho z$-half plane, in which the boundary is divided into a sequence of rods on which various linear combinations of the Killing fields $\partial_{\phi^a}$, $a=1,2$ vanish. Consider a consecutive sequence of three axis rods $\Gamma_{1}$, $\Gamma_{2}$, and $\Gamma_{3}$ separated by two corner points $p_1$ and $p_2$, as illustrated in Figure \ref{Figure 1}. We claim that the region $\Omega\subset M^4/U(1)^2$ bounded between the axes and a semi-circle connecting rod $\Gamma_1$ to $\Gamma_3$, represents a disc bundle over $S^2$. To see this observe that the middle rod $\Gamma_2$ is a 2-sphere in $M^4$. This is due to the fact that one $U(1)$ generator, say $\partial_{\phi^2}$, vanishes on $\Gamma_2$ while the other $\partial_{\phi^1}$ generates a circle at each point, except at the bounding corner points $p_1$, $p_2$ where both generators degenerate. This base $S^2$ is parameterized by the coordinate $z$ of the plane and the coordinate $\phi^1$ of the $U(1)$ generator which does not vanish on the open middle axis. The $D^2$ disc fibers may be described in the orbit space as segments emanating from $\Gamma_2$ and foliating the region $\Omega$ as shown in Figure \ref{Figure 1}. These segments represent discs over points of $S^2$. Indeed,  starting from a point on $\Gamma_2$ and fixing the coordinate $\phi^1$ on $S^2$, each point of the segments represents a circle associated to $\partial_{\phi^2}$ and this circle shrinks to a point at the starting point of the segment on $\Gamma_2$. As the foliating segments move from $\Gamma_1$ to $\Gamma_3$ the disc fibers are twisted according to rod structures. We can now transcribe Proposition \ref{prop2} to the language of rod structure.

\begin{figure}
\includegraphics{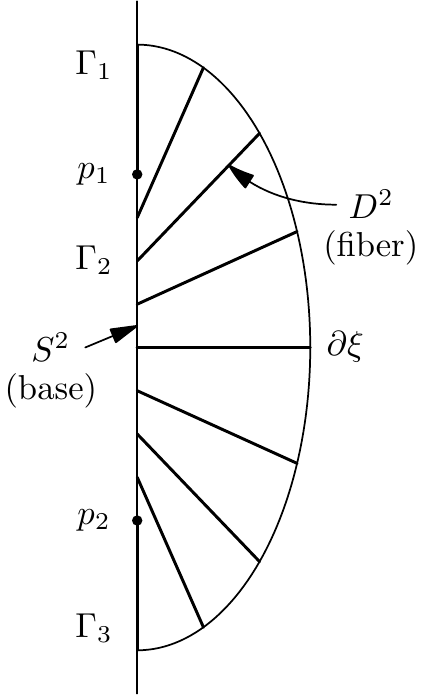}
\caption{Twisted disc bundle}  \label{Figure 1}
\end{figure}

\begin{lemma}\label{lemma1}
Consider three consecutive axis rods $\Gamma_1$, $\Gamma_2$, and $\Gamma_3$ having rod structures $(m_1,n_1)$, $(m_2, n_2)$, and $(m_3,n_3)$. The orbit space neighborhood of these rods enclosed by a semi-circle connecting $\Gamma_1$ to $\Gamma_3$ represents a disc bundle over $S^2$ in $M^4$. The zero-section self-intersection number of this disc bundle is given by
\begin{equation}
\begin{vmatrix}
m_1  & m_2     \\
 n_1  & n_2    \\
\end{vmatrix}
\begin{vmatrix}
m_2  &  m_3     \\
n_2   &  n_3 \\
\end{vmatrix}
\begin{vmatrix}
 m_3  &  m_1     \\
n_3    &  n_1  \\
\end{vmatrix}.
\end{equation}
\end{lemma}

This lemma shows that each axis rod may be interpreted as giving rise to a twisted disc bundle on $S^2$, if it is bordered on both sides by axis rods. We note that the most elementary sequence of rod structures is given by
\begin{equation}
\Gamma_1: (1,0), \qquad \Gamma_2: (0,1), \qquad \Gamma_3: (-1, k),
\end{equation}
with
\begin{equation}
-k=
\begin{vmatrix}
1  & 0    \\
 0  & 1    \\
\end{vmatrix}
\begin{vmatrix}
0  &  -1     \\
1    &  k \\
\end{vmatrix}
\begin{vmatrix}
-1  &  1     \\
k    &  0  \\
\end{vmatrix}.
\end{equation}
The rod structures $(1, 0)$, $(0, 1)$, $(-1,k)$ arise as the second columns of the $2 \times 2$ matrices appearing in equations (\ref{3-consecutive-rods-III}, \ref{3-consecutive-rods-IV}, \ref{3-consecutive-rods-I}) respectively.  Furthermore, the first column of the $2 \times 2$ matrix in (\ref{3-consecutive-rods-IV}) gives the Dehn coefficient $-k$, so that the total space of the disc bundle over $\Gamma_2$ has boundary $L(k, 1)$.

Consider now a consecutive sequence of four axis rods $\Gamma_i$, $i=1,2,3,4$. The first three rods give rise to a disc bundle $\xi_1$ on $S^2$ corresponding to a region $\Omega_1\subset M^4/U(1)^2$ between a semi-circle and the axes, and similarly the last three rods yield a disc bundle $\xi_2$ and corresponding projection $\Omega_2$ within the orbit space, see Figure \ref{Plumbing}. The region of the bundle associated with the intersection $\Omega_1\cap\Omega_2$ is homeomorphic to $B^4$ in light of \eqref{det condition}, and represents a trivialization $D_{1}^{-}\times D^2$ over the southern hemisphere of $\xi_1$ and a trivialization $D_{2}^{+}\times D^2$ over the northern hemisphere of $\xi_2$. By changing coordinates in $U(1)^2$ if necessary, we may assume without loss of generality that the rod structures for $\Gamma_{2}$, $\Gamma_3$ are $(1,0)$, $(0,1)$. Then as described above, the segments emanating from $\Gamma_3$ in Figure \ref{Plumbing} represent disc fibers in $\xi_2$ which may be given coordinates
$(r_2,\phi^2)$. Furthermore, coordinates $(r_1,\phi^1)$ may be used to parameterize the base $D_{2}^{+}$, where $r_1$ and $r_2$ are radii for the circles foliating the two discs. It follows that with respect to $\xi_2$ the region $\Omega_1\cap\Omega_2$ is parameterized by coordinates $(r_1,\phi^1,r_2, \phi^2)\in D_{2}^{+}\times D^2$. On the other hand, from the perspective of $\xi_1$ the segments emanating from $\Gamma_3$ represent sections, and are thus parameterized by the same coordinates as used for the base $(r_2,\phi^2)\in D_{1}^{-}$. Moreover the segments emanating from $\Gamma_2$ represent fibers of $\xi_1$ and are parameterized by $(r_1,\phi^1)$. This interchanging of fiber and base when passing from $\xi_1$ to $\xi_2$ is precisely the plumbing construction described in the previous section.

\begin{figure}
\includegraphics{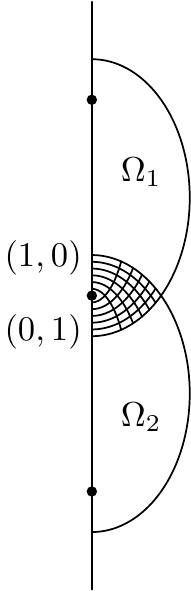}
\caption{Plumbing in the orbit space}  \label{Plumbing}
\end{figure}

\begin{lemma}\label{lemma2}
Consider four consecutive axis rods $\{\Gamma_i\}_{i=1}^{4}$ having rod structures $\{(m_i,n_i)\}_{i=1}^{4}$. The orbit space neighborhood of these rods enclosed by a semi-circle connecting $\Gamma_1$ to $\Gamma_4$ represents in $M^4$ the plumbing $\mathcal{P}(\xi_1,\xi_2)$ of the two disc bundles on $S^2$ associated with the sequences of three consecutive rods $(\Gamma_1,\Gamma_2,\Gamma_3)$ and $(\Gamma_2,\Gamma_3,\Gamma_4)$, where $\Gamma_2$ and $\Gamma_3$ represent the base $S^2$ of $\xi_1, \xi_2$.
The zero-section self-intersection numbers $-k_1$ and $-k_2$ of $\xi_1$ and $\xi_2$ determine the boundary topology of the plumbing construction through the formula $\partial\mathcal{P}(\xi_1,\xi_2)=L(k_1 k_2 -1, k_2)$.
\end{lemma}

We are now in a position to establish the decomposition of the domain of outer communication as stated in the Introduction.

\begin{proof}[Proof of Theorem \ref{thm1}]
Within the rod structure of the orbit space $M^4/U(1)^2$, let $J$ denote the number of consecutive sequences of axis rods consisting of more than two rods. Any two of the consecutive sequences are separated by either a horizon rod or a horizon puncture. Label these by $\{\Gamma_{i,j}\}_{i=0}^{I_j+1}$, $j=1,\ldots, J$ where $I_j+2$ is the length of each sequence. According to Lemma \ref{lemma1} each of the rods $\Gamma_{i,j}$, $i=1,\ldots,I_j$ gives rise to a twisted disc bundle $\xi_{i,j}$ over $S^2$. Then by repeatedly applying Lemma \ref{lemma2}, we find that each consecutive sequence of axis rods gives rise to a plumbing $\mathcal{P}(\xi_{1,j},\cdots,\xi_{I_j,j})$ of disc bundles on $S^2$ within the Cauchy surface $M^4$. Each of these plumbing constructions may be represented in the orbit space as the region bounded between a semi-circle enclosing the axis rods of the sequence, see Figure \ref{orbit}. This gives the first piece of the decomposition in \eqref{1}.

If a single axis rod is bounded on both sides by a horizon rod/puncture, then the two bounding horizon regions (indicated by shaded rectangles the figure) will be separated by a white rectangular region in the orbit space with boundary consisting of a semi-circle beginning and ending on the same axis rod. Such a domain in the orbit space corresponds in the 4-manifold to the topology $D^2\times S^1\times [0,1]$ labeled by $C^4$. These give rise to the second portion of the desired decomposition.

Consider now the $N_2$ sequences of two consecutive axis rods which are bounded on either side by a horizon rod, horizon puncture, or the asymptotic end. For each of these two rod sequences, a semi-circle in the orbit space connecting the two encloses a region which is homeomorphic to $B^4$ in $M^4$. These 4-balls make up the third piece of the decomposition \eqref{1}.

Next, portions of the orbit space semi-circles associated with the first three pieces of the decomposition \eqref{1}, together with portions of horizon semi-circles, may be connected to form a single large semi-circle $\mathcal{C}$ connecting the two semi-infinite rods and enclosing all finite rods, as shown in Figure \ref{orbit}. Within the region $\Omega$ enclosed by $\mathcal{C}$ and the $z$-axis, there are regions enclosed by semi-circles and containing the axis rods and axis punctures. These regions, which are shaded in Figure \ref{orbit}, are topologically not part of the domain of outer communication. In the 4-manifold they represent the product of an interval with a horizon cross-sectional component, and therefore removing them does not change the topology of the DOC. The complement of $\Omega$ in the orbit space may be foliated by curves homologous to $\mathcal{C}$. Since $\mathcal{C}$ represents either $S^3$, $S^1\times S^2$, or $L(p,q)$ inside $M^4$, this foliated region coincides with $M_{\text{end}}^4$ as described in Theorem \ref{thm1}. This gives the last piece of the decomposition \eqref{1}.
Lastly, formula \eqref{1.1} follows immediately from Lemma \ref{lemma1}.
\end{proof}

\begin{figure}
\includegraphics{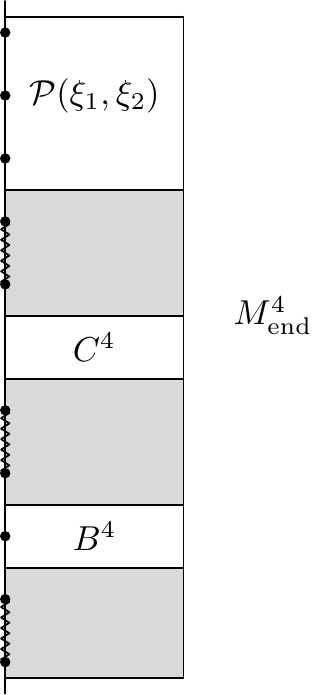}
\caption{Decomposition of orbit space}  \label{orbit}
\end{figure}

We now seek to find a canonical way to fill in the horizons and cap off the infinity by appropriate compact simply connected 4-manifolds with a single component boundary. There are three possible boundary types that are needed for this procedure, namely the sphere $S^3$, the ring $S^1\times S^2$, and a lens $L(p,q)$. Since all three arise via plumbing of disc bundles on $S^2$, and the plumbing construction has a naturally associated rod structure, we are motivated to take this approach. Consider a horizon rod or puncture which is bounded between two axis rods having rod structures $(m,n)$ and $(u,v)$. By applying a $SL(2,\mathbb{Z})$ transformation, that is a change of coordinates in $U(1)^2$, we may assume that $(m,n)=(1,0)$ and $(u,v)=(-q,p)$. The horizon then has the topology of $L(p,q)$. If $q=0$ then this is a ring $S^1\times S^2$ and the fill-in is simply $S^2\times D^2$, which is the trivial disc bundle over $S^2$. So assume that $p > q\neq 0$ and compute the continued fraction
\begin{equation}
\frac{p}{q} = k_1- \cfrac{1}{k_2-\cfrac{1}{k_3-\cfrac{1}{{\cdots} - \frac{1}{k_\ell}}}}.
\end{equation}
with $k_i> 0$.
Each integer $-k_i$ represents the zero-section self-intersection number of a disc bundle $\xi_i$ on $S^2$, and according to Proposition \ref{prop2} these may be plumbed together to form a compact simply connected 4-manifold $\mathcal{P}(\xi_1,\cdots,\xi_{\ell})$ having a single component boundary of topology $L(p,q)$. In the notation of Section \ref{sec2},  setting $\partial_t=\partial_{\phi^1}$ and $\partial_s=\partial_{\phi^2}$ shows that each disc bundle $\xi_i$ is associated to rod $\Gamma_i$ in the orbit space $\mathcal{P}(\xi_1,\cdots,\xi_{\ell})/U(1)^2$ having a rod structure determined by the self-intersection numbers. In particular, we obtain a sequence of rod structures
\begin{equation}
(1,0), (0,1), (-1, k_1), (-k_2, k_1k_2-1), \cdots , (q, p),
\end{equation}
where $\xi_1$ is paired with rod structure $(0,1)$, $\xi_2$ is paired with $(-1,k_1)$ and so on. Since the first and last rod structure agree with those bounding the original horizon rod/puncture, this sequence of rod structures may be inserted in place of the horizon rod/puncture to create an expanded version of the rod structures for the domain of outer communication. This means that the $U(1)^2$-action associated with the plumbing construction coincides with the symmetry action on the horizon cross-section. This process is equivalent to gluing the plumbing construction in to fill the horizon, or alternatively filling in the shaded regions in Figure \ref{orbit}.

The same process of filling in a horizon may also be applied to capping off the asymptotic end. The two semi-infinite rods within the rod structure of the DOC play the role of axis rods which bound a horizon rod/puncture. In particular, for an asymptotically flat end represented by $(1, 0)$ and $(0, 1)$ semi-infinite rods, the act of capping is described in Section \ref{spherical-fill-ins}.

\begin{lemma}\label{lemma3}
For each horizon rod, horizon puncture, or asymptotic end, there exists a natural choice of a compact simply connected 4-manifold with single component boundary which fills in the horizon or caps off the infinity.
\end{lemma}

We are now in a position to establish the classification of the compactified Cauchy surface.

\begin{proof}[Proof of Theorem \ref{thm2}]
Consider the Cauchy surface $M^4$ for the domain of outer communication $\mathcal{M}^5$. By Lemma \ref{lemma3} there exist simply connected fill-ins $\tilde{M}^4_h$ for the horizons, and a simply connected cap $\tilde{M}^4_{\text{end}}$. van Kampen's theorem shows that after inserting the fill-ins and cap, the resulting compactified manifold $\tilde{M}^4$ is simply connected. Moreover according to the construction of the fill-ins and caps, $\tilde{M}^4$ comes equipped with an effective $U(1)^2$-action. According to \cite{OrlikRaymond1} the orbit space $\tilde{M}^4/U(1)^2$ is a 2-dimensional disc, such that the boundary circle is divided into a sequence of rods with rod structures detailing how the action degenerates. This sequence of rod structures corresponds to that of $M^4/U(1)^2$, with additional rods added in place of horizon rods/punctures and the asymptotic end which may be computed from the proof of Lemma \ref{lemma3}. Furthermore the results (pages 553 and 554) of \cite{OrlikRaymond1} show that $M^4$ must then be either $S^4$, or a finite connected sum of $S^2\times S^2$, $\mathbb{CP}^2$, and $\overline{\mathbb{CP}}^2$. Since $\mathbb{CP}^2 \# S^2\times S^2 \cong \mathbb{CP}^2 \# \overline{\mathbb{CP}}^2 \# \mathbb{CP}^2$, the connected sum decomposition of $\tilde{M}^4$ may be expressed solely in terms of $S^2\times S^2$ or in terms of $\mathbb{CP}^2$ and $\overline{\mathbb{CP}}^2$.

An alternative approach to obtaining this connected sum decomposition of $\tilde{M}^4$ is to apply the classification theorem of Freedman \cite{FreedmanQuinn} and work of Donaldson \cite{DonaldsonKronheimer}. The desired result follows immediately, except for the possibility of having the $E_8$-manifold present as a component in the connected sum. However, such components can be ruled out as in \cite{HollandsHollandIshibashi,HollandsYazadjiev1}.
\end{proof}

\section{Examples}
\label{sec4} \setcounter{equation}{0}
\setcounter{section}{4}

\begin{figure}
\centering
  \hspace*{\fill}%
  \subcaptionbox{~}{\includegraphics{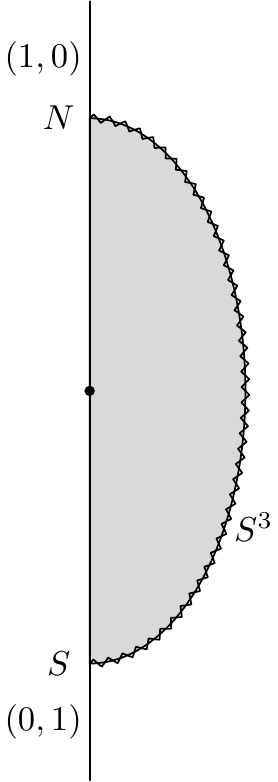}}\hspace{\fill}%
  \subcaptionbox{~}{\includegraphics{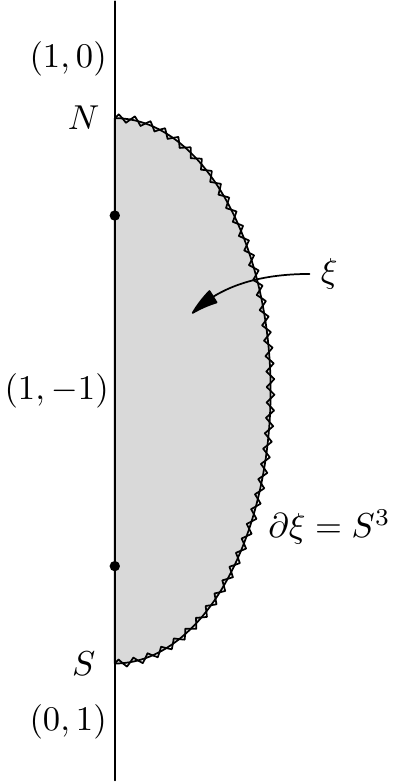}}\hspace{\fill}%
  \subcaptionbox{~}{\includegraphics{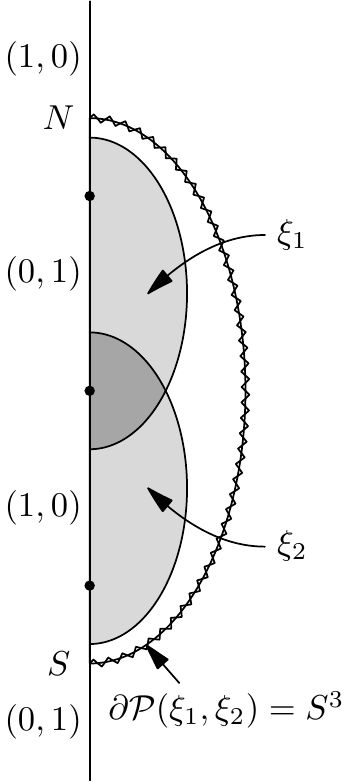}}%
  \hspace*{\fill}%
  \caption{Three fill-ins}\label{fig5}
\end{figure}

In this section we consider basic examples of domains of outer communication having horizons of three different topological types. The methods of Theorems \ref{thm1} and \ref{thm2} are applied to classify the DOCs when an asymptotically flat end is present. In addition, we show that other methods for filling in the horizon produce different topologies for the compactified Cauchy slice. The choices of fill-in made in this paper may be described as canonical in the sense that they are systematized, and offer the most elementary option which is simply connected.

\subsection{Spherical Horizon}\label{spherical-fill-ins}
Consider an asymptotically flat DOC with $S^3$ horizon and having the typical rod structure $\{(1,0),(0,0),(0,1)\}$. An example of such a vacuum black hole is given by the Myers-Perry solution \cite{MyersPerry}.
The horizon fill-in is given by $\tilde{M}^4_h=B^4$. In terms of the rod structure this fill-in entails removing the horizon rod $(0,0)$ to obtain $\{(1,0),(0,1)\}$. The cap at infinity is the same, namely $\tilde{M}^{4}_{\text{end}}=B^4$. This yields the compactified manifold $\tilde{M}^4= S^4$. The Cauchy surface of the DOC is then $M^4=\left(S^4\setminus\cup_{i=1}^2 B_i^4\right)\cup M^4_{\text{end}}=\left(S^4\setminus B^4 \right) \#\mathbb{R}^4$. See Figure \ref{fig5}~(a). Here as in all figures to follow, squiggly curves represent the horizon.

A noncanonical fill-in for the $S^3$ horizon is to use the twisted disc bundle $\xi$ with self-intersection number $-1$, as shown in Figure \ref{fig5}~(b).

In Figure \ref{fig5}~(c) another noncanonical possibility is displayed in which the horizon is filled in with $\tilde{M}^4_h=\mathcal{P}(\xi_1,\xi_2)$, the plumbing of two trivial disc bundles over $S^2$. Recall that according to the discussion in Section \ref{sec2} the boundary $\partial\mathcal{P}(\xi_1,\xi_2)=S^3$. This entails replacing the horizon rod with the sequence of rod structures $(0,1),(1,0)$, to obtain the expanded or enhanced rod structure $\{(1,0),(0,1),(1,0),(0,1)\}$. The compactified Cauchy surface of the DOC is then $\tilde{M}^4= S^2\times S^2$, which may be computed from the chart in \cite[pg. 552]{OrlikRaymond1}. In this case the Cauchy slice of the DOC has topology $M^4=\left(S^2\times S^2\setminus\mathcal{P}(\xi_1,\xi_2)\right)\#\mathbb{R}^4$.

\subsection{Ring Horizon}

\begin{figure}
\centering
  \hspace*{\fill}%
  \subcaptionbox{~}{\includegraphics{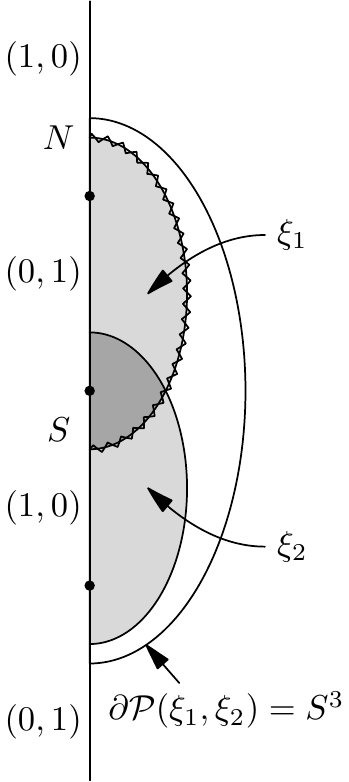}}\hfill%
  \subcaptionbox{~}{\includegraphics[align=c,width=.5\textwidth]{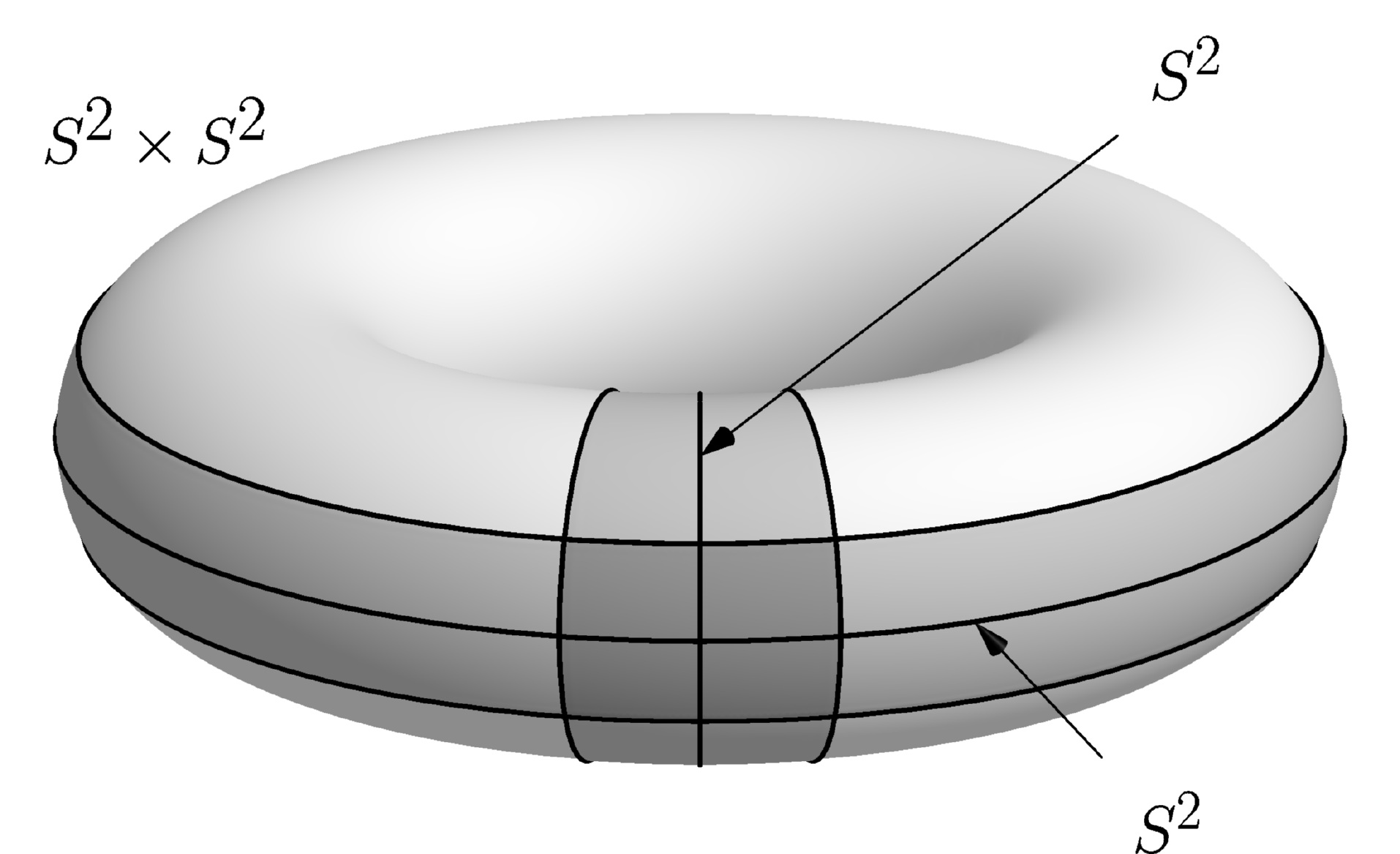}}
  \hspace*{\fill}%
  \caption{Standard fill-in for black ring} \label{fig:ring1}
\end{figure}

Consider an asymptotically flat black ring having the set of rod structures
\begin{equation}
\{(1,0),(0,0),(1,0),(0,1)\}.
\end{equation}
Explicit solutions to the vacuum equations having this rod structure are given by the black rings of Emparan-Reall and Pomeransky-Senkov \cite{EmparanReall,PomeranskySenkov}.
Following the prescription of Theorem \ref{thm2}, the horizon fill-in is $\tilde{M}^4_h=S^2\times D^2$. With regards to the rod structure, this fill-in requires replacing the horizon rod structure with a rod of structure $(0,1)$ to obtain the expanded rod structure
\begin{equation}
\{(1,0),(0,1),(1,0),(0,1)\}.
\end{equation}
At infinity the cap is again $\tilde{M}^4_{\text{end}}=B^4$. The compactified manifold $\tilde{M}^4$ then has a disc orbit space with rod structure given by the expanded sequence, and this corresponds to $\tilde{M}^4= S^2\times S^2$. Therefore the Cauchy slice of the DOC is $M^4=\left(S^2\times S^2 \setminus S^2\times D^2\right)\# \mathbb{R}^4$. See Figure \ref{fig:ring1}.

There is an alternative noncanonical way to fill in the ring horizon. Namely, choose the fill-in to be $\tilde{M}^4_1=S^1\times D^3$, which is not simply connected. This fill-in has previously been examined in \cite[pg. 652]{HollandsYazadjiev1}, and the compactified Cauchy slice is then $\tilde{M}^4=S^4$. From this we find that the Cauchy slice of the DOC has the topology $M^4=\left(S^4\setminus S^1\times D^3\right) \#\mathbb{R}^4=S^2\times D^2 \#\mathbb{R}^4$. See Figure \ref{fig:ring2}. Note that in~(b) of Figure~\ref{fig:ring1} the 2-dimensional torus represents $S^2\times S^2$ with two dimensions suppressed, and the ambient $\R^3$ is to be ignored. On the other hand, in~(b) of Figure~\ref{fig:ring2} the solid torus represents $S^1\times D^3$, with one dimension along the vertical $\R^2$ axis suppressed, and the ambient space is the DOC, i.e.\ after compactification $S^4$.

\begin{figure}
\centering
  \hspace*{\fill}%
  \subcaptionbox{~}{\includegraphics{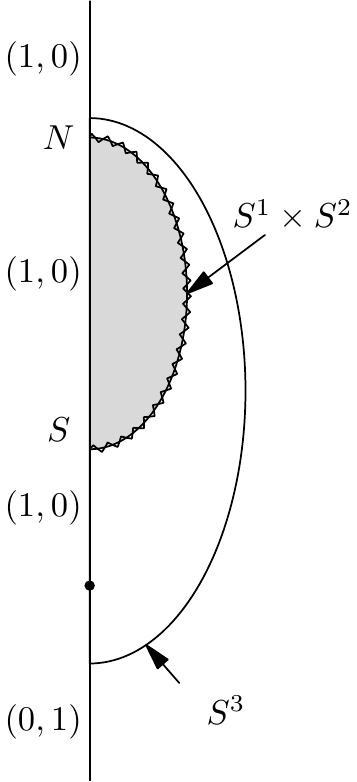}}\hfill%
  \subcaptionbox{~}{\includegraphics[align=c,width=.56\textwidth]{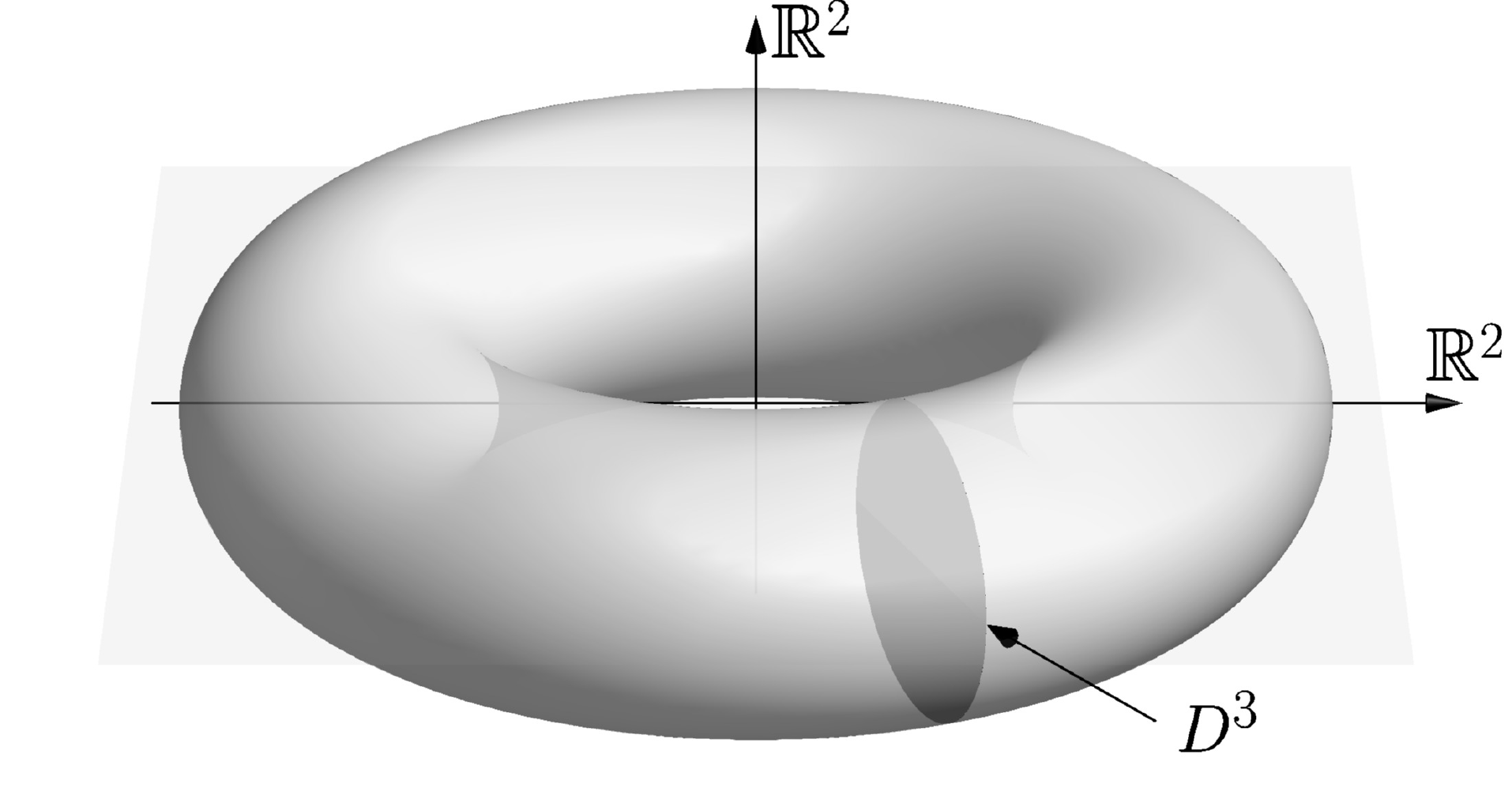}}%
  \hspace*{\fill}%
\caption{Non standard fill-in for black ring}  \label{fig:ring2}
\end{figure}

\subsection{Lens Horizon}
\label{lh}

Consider now the asymptotically flat black lens having rod structures
\begin{equation}
\{(1,0),(0,0),(-1,p),(0,1)\}.
\end{equation}
Following the proof of Theorem \ref{thm2} leads to the horizon fill-in $\tilde{M}^4_h=\xi$, that is the twisted disc bundle over $S^2$ with zero-section self-intersection number $-p$.  This entails replacing the horizon rod with a rod of structure $(0,1)$ to obtain the expanded rod structure
\begin{equation}
\{(1,0),(0,1),(-1,p),(0,1)\},
\end{equation}
and as before the cap at infinity is $\tilde{M}^4_{\text{end}}=B^4$. The compactified manifold $\tilde{M}^4$ then has a disc orbit space with rod structure given by the expanded sequence. A computation \cite{OrlikRaymond1} shows that this corresponds to
\begin{equation}
\tilde{M}^4=
\begin{cases}
      S^2\times S^2 & p=\text{even}, \\
      \mathbb{CP}^2\#\overline{\mathbb{CP}}^2 & p=\text{odd}>1.
   \end{cases}
\end{equation}
Therefore the Cauchy slice of the DOC is given by $M^4=\left(\tilde{M}^4 \setminus \xi\right)\# \mathbb{R}^4$. In particular, the black lens ($\mathbb{RP}^3$) solution of 5D minimal supergravity constructed by Kunduri and Lucietti \cite{LuciettiKunduri} has the Cauchy slice topology $\left(S^2\times S^2 \setminus \xi\right)\# \mathbb{R}^4$.

\subsection{Multiple Black Holes}

In 5-dimensions it is possible to have stationary vacuum black holes with multi-component horizons. For instance, the black Saturn solution
\cite{ElvangFigueras} has an $S^3$ horizon component surrounded by a ring $S^1\times S^2$ component. The associated rod structure is given by
\begin{equation}
\{(1,0),(0,0),(0,1),(0,0),(0,1)\}.
\end{equation}
In order to compactify the DOC following the procedure outlined above, we may use a ball $B^4$ for the spherical component and the trivial disc bundle $D^2\times S^2$ for the ring component. The resulting extended rod structure then becomes
\begin{equation}
\{(1,0),(0,1),(1,0),(0,1)\},
\end{equation}
which corresponds to $\tilde{M}^4=S^2\times S^2$. Thus the topology of a Cauchy slice of the DOC for the black Saturn solution is $\left(S^2\times S^2 \setminus (B^4\cup D^2\times S^2)\right)\# \mathbb{R}^4$.

Another asymptotically flat multi-black hole solution of the vacuum equations involves two concentric singly spinning rings rotating in the same plane. This is the so-called dipole black ring (or di-rings) constructed in \cite{EvslinKrishnan,IguchiMishima}. Its rod structure is
\begin{equation}
\{(1,0),(0,0),(1,0),(0,1),(0,0),(0,1)\}.
\end{equation}
By filling in the two ring horizons with the trivial disc bundle over $S^2$ the resulting extended rod structure sequence takes the form
\begin{equation}
\{(1,0),(0,1),(1,0),(0,1),(1,0),(0,1)\},
\end{equation}
which gives rise to the compactified manifold $\tilde{M}^4=S^2\times S^2\# S^2\times S^2$. Hence the Cauchy slice toplogy of the DOC is $\left(S^2\times S^2\# S^2\times S^2 \setminus (D^2\times S^2\cup D^2\times S^2)\right)\# \mathbb{R}^4$.

\subsection{Nonuniqueness of DOC}

At the end of Section \ref{sec1}, we brought up the question of whether the DOC is uniquely determined by the horizon topology and the topology of the asymptotic end. To illustrate the negative answer to this question, here examples of two different asymptotically flat DOCs will be given, both of which have a single component horizon cross-section with $\mathbb{RP}^3$ topology. In order to describe the bi-axisymmetric solutions to the vacuum equations, it suffices to provide the sequences of rod structures associated with the orbit space, namely
\begin{equation}\label{r1}
\{(1,0),(0,0),(-1,2),(0,1)\},
\end{equation}
\begin{equation}\label{r2}
\{(1,0),(0,0),(-1,2),(0,1),(1,0), (0,1)\}.
\end{equation}
Observe that both sets of rod structures begin and end with $(1,0),(0,1)$ indicating that the asymptotic end is of the form $\mathbb{R}_{+}\times S^3$, and both horizon rods are bounded between the axis rods $(1,0),(-1,2)$ signifying that the horizon topology is the lens space $L(2,1)=\mathbb{RP}^3$. The only difference between the two sequences is that addition of two axis rods in \eqref{r2} having rod structures $(0,1),(1,0)$. This adds two additional corners and changes the topology of the corresponding DOCs.

To see the differing topologies, fill in the horizon as in Section \ref{lh} with the twisted disc bundle $\tilde{M}^4_h=\xi$ over $S^2$ having zero-section self-intersection number $-2$. In terms of the rod structures this is equivalent to replacing the horizon rod structure with the rod structure $(0,1)$ to obtain the enhanced sequences
\begin{equation}
\{(1,0),(0,1),(-1,2),(0,1)\},
\end{equation}
\begin{equation}
\{(1,0),(0,1),(-1,2),(0,1),(1,0), (0,1)\}.
\end{equation}
After capping off the end the resulting compactified manifolds $\tilde{M}^4$
have topology $S^2\times S^2$ and $S^2\times S^2 \# S^2\times S^2$, respectively.
Therefore the two domains of outer communication are not homeomorphic.

\end{document}